\documentclass[conference]{IEEEtran}
\usepackage{multirow}
\usepackage{amsmath}
\usepackage[table]{xcolor}
\usepackage{graphicx}
\usepackage{amssymb}
\usepackage{wrapfig}
\usepackage{hyperref}
\usepackage{pbox}
\usepackage{paralist}
\usepackage{hhline}

\newtheorem{definition}{Definition}
\newtheorem{lemma}{Lemma}

\newtheorem{theorem}{Theorem}
\newtheorem{proof}{IEEEproof}
\newtheorem{corollary}{Corollary}
\newtheorem{example}{Example}

\def\olg#1{\mathit{Olg}(#1)}

\def\phi{\varphi}
\def\theta{\vartheta}
\newcommand{\tran}[1]{\xrightarrow[]{#1}}

\def\ff{\mathit{ff}}
\def\tt{\mathit{tt}}

\def\Polsat{{Polsat}}
\def\Aalta{{Aalta}}
\def\expand{\mathit{NF}}
\def\of#1{\mathit{of}(#1)}
\def\ofp#1{\mathit{ofp}(#1)}
\def\nondeter{\perp}
\def\cur{\mathit{cur}}
\def\inf{\mathit{inf}}

\title{Fast LTL Satisfiability Checking by SAT Solvers}

\author{
    \IEEEauthorblockN{Jianwen Li\IEEEauthorrefmark{1}, Geguang Pu\IEEEauthorrefmark{1}, Lijun Zhang\IEEEauthorrefmark{2},
    Moshe Y. Vardi\IEEEauthorrefmark{3} and Jifeng He\IEEEauthorrefmark{1}}
    \IEEEauthorblockA{\IEEEauthorrefmark{1}Software Engineering, East China Normal University
   }
    \IEEEauthorblockA{\IEEEauthorrefmark{2}State Key Laboratory of Computer Science, Institute of Software, Chinese Academy of Sciences 
   }
    \IEEEauthorblockA{\IEEEauthorrefmark{3}Computer Science, Rice University
        }
}

\begin{document}

\maketitle

\begin{abstract}
Satisfiability checking for Linear Temporal Logic (LTL) is a fundamental 
step in checking for possible errors in LTL assertions. Extant LTL
satisfiability checkers use a variety of different search procedures. With 
the sole exception of LTL satisfiability checking based on bounded model
checking, which does not provide a complete decision procedure,
LTL satisfiability checkers have not taken advantage of the remarkable 
progress over the past 20 years in Boolean satisfiability solving.
In this paper, we propose a new LTL satisfiability-checking framework that 
is accelerated using a Boolean SAT solver.  Our approach is based on the 
variant of the \emph{obligation-set method}, which we proposed in
earlier work. We describe here heuristics that allow the use of a
Boolean SAT solver to analyze the obligations for a given LTL formula. 
The experimental evaluation indicates that the new approach provides
a significant performance advantage.
\end{abstract}

\section{Introduction}

The \emph{satisfiability problem} for Linear Temporal Logic (LTL)
asks whether a given LTL formula is satisfiable \cite{SC85}.  LTL 
satisfiability checking plays an important role in checking the 
consistency of linear temporal specifications that are often used 
in an early stage of system design \cite{RV10,RV11}.  Thus, efficient 
decision procedures to reason about large LTL formulas are quite 
desirable in practice.  

There have been several approaches proposed to deal with the LTL
satisfiability checking problem.  The model-checking 
approach reduces LTL satisfiability to LTL model checking by model checking
the \emph{negation} of the given formula against a \emph{universal} model.
This approach uses either explicit \cite{RV10} or symbolic \cite{RV11}
model checking.  The tableau-based \cite{Sch98} and antichain-based
\cite{DDMR08} approaches apply an on-the-fly search in the underlying 
automaton transition system.  The temporal-resolution-based method 
explores the unsatisfiable core using a deductive system \cite{HK03}. 
Our own previous work \cite{LZPVH13}, embodied in the \Aalta\ LTL 
satisfiability checker, follows the automata-based approach and reduces 
satisfiability checking to emptiness checking of the transition system 
by adopting two new heuristic techniques, using \emph{on-the-fly search} and 
\emph{obligation sets}.

Previous experimental evaluations across a wide spectrum of benchmarks 
\cite{RV10,RV11,SD11} concluded that none of existing approaches
described above dominate others. To establish a high-performance LTL 
satisfiability checker, we introduced a \emph{portfolio} LTL solver named 
\Polsat\ \cite{LPZVHCoRR13}, which runs several approaches in parallel,
terminating with the fastest thread. By definition, \Polsat\ is the
best-performing LTL satisfiability checker (subject to
constraints on the number of parallel threads).

An interesting observation in \cite{LPZVHCoRR13} is that the 
bounded-model-checking (BMC) technique \cite{CBRZ01} is the fastest 
on satisfiable formulas, as it leverages the tremendous progress 
demonstrated by Boolean satisfiability (SAT) solvers over the last 20 years 
\cite{MZ09}.  At the same time, BMC can detect satisfiability, but not 
unsatisfiability, which means that this approach does not provide a 
complete decision procedure. Nevertheless, the impressive performance of the
BMC-based approach inspired us to explore other possibilities of leveraging 
SAT solvers in LTL satisfiability checking.

We propose here an LTL satisfiability-checking framework that can be greatly
accelerated by using SAT solvers. The key idea here is of using
\emph{obligation formulas}, which are Boolean formulas
collecting satisfaction information from the original LTL
formula.  Intuitively, an LTL formula is satisfiable if the
corresponding Boolean \emph{obligation formula} is satisfiable. 
Using obligation formulas makes it possible to utilize SAT solving, 
since it eliminates the temporal information of LTL formula. Based on
obligation formulas, we extend the approach proposed in \cite{LZPVH13}
by presenting two novel techniques to accelerate satisfiability
checking procedure with SAT solvers.  In contrast to the BMC-based approach,
our method is both sound and complete, as it can also check unsatisfiable 
formulas.

To illustrate the efficiency of our new approach, we integrate our
implementation, \Aalta\_v0.2 into \Polsat, which also provides a
testing environment for LTL solvers.  The experiments show that 
while still no solver dominates across all benchmarks,  \Aalta\_v0.2
is much more competitive with other LTL satisfiability checkers than
\Aalta\_v0.1. More significantly, the performance of \Polsat\ improves 
dramatically as a result of replacing  \Aalta\_v0.1 by  \Aalta\_v0.2.

\noindent
\textbf{Contributions}: The three main contributions of the paper are as 
follows: 1) We extend the concept of \emph{obligation set} to that of
\textit{obligation formulas}, which enables us to leverage Boolean
satisfiability solving in LTL satisfiability solving.
2) We offer two novel SAT-based heuristics to boost the checking of 
satisfiable and unsatisfiable formulas respectively. 
3) We present a new tool, \Aalta\_v0.2, which is integrated into \Polsat\ 
and evaluated over large set of benchmark formulas. The experiments show 
that the new approach is both effective and efficient: the performance of
\Polsat\ improves 10-fold in some cases, and an average of 30\% to 60\% 
speed-up on random formulas.

\noindent
\textbf{Paper Structure}: The paper is organized as
follows. Section \ref{sec:pre} introduces the preliminaries about LTL and our
previous work~\cite{LZPVH13}. Section \ref{sec:of} provides the theoretical
framework of this paper. In Section \ref{sec:ofa}, we describe two techniques, 
based on SAT solving, to accelerate satisfiability checking respectively for
satisfiable and unsatisfiable formulas. The empirical framework
is described in Section \ref{sec:exp}. Section \ref{sec:related} discusses 
related work, and finally Section \ref{sec:con} concludes the paper.%

\section{Preliminaries}\label{sec:pre}
\subsection{Linear Temporal Logic}\label{sec:ltl}
Let $AP$ be a set of atomic properties. 
The syntax of LTL formulas is defined by:
\begin{align*}
\phi\ ::=\ \tt \mid \ff\mid a \mid \neg \phi \ |\ \phi\wedge \phi\ |\ \phi\vee \phi\ |\ \phi U \phi\ |\ \phi R\ \phi\ |\ X \phi
\end{align*}
where $a\in AP$, $\phi$ is an LTL formula.  We use the usual
abbreviations: 
$Fa=\tt Ua$, and $Ga=\ff Ra$. 

We say $\phi$ is a propositional formula if it does not contain
temporal operators. We say $\phi$ is a \emph{literal}
if it is an atomic proposition or its negation.  We use $L$ to denote
the set of literals, lower case letters $a,b,c,l$ to denote
literals, $\alpha$ to denote propositional
formulas, and $\phi, \psi$ for LTL formulas.  
In this paper, we consider LTL formulas in negation normal form (NNF) -- all negations are 
pushed in front of atomics. 
LTL formulas are often interpreted over $(2^{AP})^\omega$. Since we consider LTL in NNF, formulas are interpreted on infinite literal sequences
$\Sigma:= (2^L)^\omega$.

A \emph{trace} $\xi=\omega_0\omega_1\omega_2\ldots$ is an infinite sequence over
$\Sigma^\omega$. For $\xi$ and $k\geq 1$ we use $\xi ^k=\omega_0\omega_1\ldots
\omega_{k-1}$ to denote the prefix of $\xi$ up to its $k$-th element,
and $\xi _k =\omega_k\omega_{k+1}\ldots$ to denote the suffix of $\xi$
from its $(k+1)$-th element.  Thus, $\xi=\xi^k \xi_k$. 
The semantics of temporal operators with respect to an infinite trace $\xi$ is 
given by: $\xi\models \alpha$ iff $\xi^1\models\alpha$;
$\xi \models X\ \phi$ iff $\xi_1\models \phi$; and
\begin{itemize}
\item $\xi \models \phi_1\ U\ \phi_2$ iff there exists $i\geqslant 0$
such that $\xi_i\models \phi_2$ and for all $0 \leqslant j < i,
\xi_j\models \phi_1$;
\item $\xi \models \phi_1\ R\ \phi_2$ iff either
  $\xi_i\vDash\phi_2$ for all $i\geq 0$, or there exists $i\ge 0$
  with $\xi_i\models\phi_1\wedge\phi_2$ and $\xi_j\vDash\phi_2$ for
  all $0\leq j< i$.
\end{itemize}
According to the semantics, it holds $\phi R\psi=\neg (\neg \phi U\neg
\psi)$. Now we define the satisfiability of LTL formulas as follows:

\begin{definition}[Satisfiability]
We say $\phi$ is satisfiable if there
exists an infinite trace $\xi$ such that $\xi\models\phi$.
\end{definition}

\subsection{Obligation-Based Satisfiability Checking}

This section recalls the fundamental theories on obligation-based satisfiability 
checking in our previous work \cite{LZPVH13}. For more details readers can refer to the literature.

\noindent
\textbf{Obligation Set} The \textit{obligation set} defined below is the fundamental part of the generalized satisfiability checking in our previous work.

\begin{definition}[Obligation Set]\label{def:os}
  For a formula $\phi$, we define its obligation set, denoted by
  $\olg{\phi}$, as follows:
  \begin{itemize}
  \item $\olg{\tt}=\{\emptyset\}$ and $\olg{\ff}=\{\{\ff\}\}$;
   \item If $\phi$ is a literal, $\olg{\phi}=\{ \{\phi\}\}$;
    \item If $\phi=X\psi$, $\olg{\phi}=\olg{\psi}$;
    \item If $\phi=\psi_1\vee\psi_2$, $\olg{\phi}=\olg{\psi_1}\cup \olg{\psi_2}$;
    \item If $\phi=\psi_1\wedge\psi_2$, $\olg{\phi}=\{O_1\cup O_2 \mid O_1\in \olg{\psi_1}\wedge O_2\in \olg{\psi_2}\}$;
    \item If $\phi=\psi_1 U\psi_2$ or $\psi_1 R\psi_2$, $\olg{\phi}=\olg{\psi_2}$.
 \end{itemize}
 For $O\in \olg{\phi}$, we refer to it as an \emph{obligation} of
 $\phi$. Moreover, we say $O$ is a consistent obligation iff $\bigwedge a\not\equiv\ff$ holds, 
 where $a\in O$.
\end{definition}

From the definition of \textit{obligation} above, one can check easily the following theorem 
is true:
\begin{theorem}[Obligation Acceleration \cite{LZPVH13}]\label{thm:oa}
  Assume $O\in \olg{\phi}$ is a consistent obligation. Then,
  $O^\omega\models \phi$.
\end{theorem}

\noindent
\textbf{Obligation-Based Satisfiability Checking}
Theorem \ref{thm:oa} is sound but not complete. If no consistent
obligations are found, we shall then explore the \textit{LTL
  Transition System}, which uses the \textit{Normal Form} defined as
follows:

\begin{definition}[Normal Form]\label{def:expansion}
  The \textit{normal form} of an LTL formula $\phi$, denoted as
  $\expand(\phi)$, is a set defined as follows:
\begin{compactenum}
\item $\expand(\phi) =\{\phi \wedge X(\tt)\}$ if $\phi\not=\ff$ is a
  propositional formula. If $\phi =\ff$, we define $\expand(\ff)=\emptyset$;
\item $\expand(X\phi) = \{\tt\wedge X(\psi)\mid \psi\in DF(\phi)\}$;
\item $\expand(\phi_1 U \phi_2) = \expand(\phi_2)\cup \expand( \phi_1 \wedge
  X(\phi_1 U \phi_2))$;
\item $\expand(\phi_1 R \phi_2) = \expand(\phi_1 \wedge \phi_2) \cup \expand( \phi_2 \wedge X(\phi_1 R \phi_2))$;
\item $\expand(\phi_1 \vee \phi_2) = \expand(\phi_1)\cup \expand(\phi_2)$;
\item $\expand(\phi_1\wedge\phi_2) = \{(\alpha_1\wedge\alpha_2) \wedge X(\psi_1\wedge\psi_2)\mid \forall i=\mathit{1,2}. \ \alpha_i\wedge X(\psi_i)\in \expand(\phi_i)\}$.
\end{compactenum}
\end{definition}

Note here let $\phi = \bigvee_{i\in I} \phi_i$ such that the root
operator of $\phi_i$ is not a disjunction, and then $DF(\phi):=\{\phi_i \mid i\in I\}$ 
is defined as the set of
disjuncts of $\phi$. Now we introduce the \textit{LTL transition system}:

\begin{definition}[LTL Transition System]\label{def:lts}
  Let $\phi$ be the input formula. The
  labeled transition system $T_{\phi}$ is a tuple $\langle Act,
  S_\phi, \tran{}, \phi\rangle$ where: 

  \begin{compactenum}
  \item  $\phi$ is the initial state, 
  \item $Act$ is the set of conjunctive formulas over $L_{\phi}$,
 \item the transition relation $\tran{}\ \subseteq S_\phi\times Act\times S_\phi$ is defined by:
    $\psi_1\tran{\alpha}\psi_2$ iff there exists $\alpha\wedge
    X(\psi_2) \in \expand(\psi_1)$,
\item $S_\phi$ is the smallest set of formulas such that $\psi_1\in
  S_\phi$, and $\psi_1\tran{\alpha}\psi_2$ implies
  $\psi_2\in S_\phi$.
  \end{compactenum}
\end{definition}

For a strong connected component (SCC) $scc$, we use $L(scc)$ to denote the set of literals that along with $scc$. 
Then we have the following theorem:

\begin{theorem}[Obligation-Based Satisfiability Checking \cite{LZPVH13}]\label{thm:scc}
  The formula $\phi$ is satisfiable iff there exists a SCC
  $scc$ of $T_\phi$ and a state $\psi$ in $scc$ such that $L(scc)$ is a
  superset of some obligation $O\in \olg{\psi}$.
\end{theorem}

\begin{example}
  \begin{enumerate}
    \item Consider the formula $\phi=a U b \wedge c U d$: Since $\olg{\phi}=\{\{b, d\}\}$  in 
    which $\{b,d\}$ is 
    obviously a consistent obligation, so $\phi$ is also satisfiable from Theorem \ref{thm:oa}.
    \item Consider the formula $\phi=F a \wedge G\neg a$: since 
    $\olg{\phi}=\{\{a,\neg a\}\}$ which does not contain any consistent obligation, so Theorem 
    \ref{thm:oa} is not available. Actually, $T_\phi$ contains only one state $\phi$ with a 
    self-loop labeling $\neg a$: Thus we cannot find a $scc$ satisfying Theorem \ref{thm:scc}, 
    which implies $\phi$ is unsatisfiable. 
  \end{enumerate}
\end{example}

\section{Satisfiability Checking with Obligation Formula}\label{sec:of}

\subsection{Obligation Formula}

The \textit{Obligation}-based satisfiability checking has been proven more efficiently than 
traditional model-checking-based approach \cite{LZPVH13}. However the size of obligation set 
can be exponential in the number of conjuncts. 
For example, consider the pattern formula 
$\bigwedge_{1\leq i\leq n} (G a_i\vee F b_i)$, which obviously is satisfiable. 
By applying our previous approach, the extra exponential cost must be paid to compute the 
whole \textit{obligation set}. We may view the \textit{obligation set} 
as a DNF, with  each element in \textit{obligation set} a clause in DNF. 
It hints that we can replace the \textit{obligation set} by an \textit{obligation formula}.

\begin{definition}[Obligation Formula]\label{def:of}
Given an LTL formula $\phi$, the corresponding \textit{obligation formula}, 
which is denoted as $\of{\phi}$, is defined recursively as follows:
     \begin{itemize}
       \item $\of{\tt}=\tt$ and $\of{\ff}=\ff$;
       \item If $\phi = p$ where $p$ is a literal, then $\of{\phi} = p$;
       \item If $\phi = X\psi$, then $\of{\phi} = \of{\psi}$;
       \item If $\phi = \phi_1 U \phi_2$ or $\phi = \phi_1 R \phi_2$, then $\of{\phi} = \of{\phi_2}$;
       \item If $\phi = \phi_1 \wedge \phi_2$, then $\of{\phi} = \of{\phi_1} \wedge \of{\phi_2}$;
       \item If $\phi = \phi_1 \vee \phi_2$, then $\of{\phi} = \of{\phi_1} \vee \of{\phi_2}$;
     \end{itemize}
\end{definition}

The \textit{obligation formula} is virtually a Boolean formula.
Compared to the definition of \textit{obligation set} (Definition
\ref{def:os}), the \textit{obligation formulas} avoid the generation
of DNF, and thus avoid the extra exponential cost. It succeeds to
reduce the computation of \textit{obligation set} to the checking on
the \textit{obligation formula}.

The following lemma explains the relationship between the \textit{obligation formula} and \textit{obligation set}:

\begin{lemma}\label{lemma:ofandos}
  Given an LTL formula $\phi$, then
  $\of{\phi}\equiv\bigvee_{O\in\olg{\phi}}\bigwedge_{l\in O} l$,
  i.e. the DNF of $\of{\phi}$ is
  $\olg{\phi}$.
\end{lemma}
\begin{proof}
  We can prove this lemma by structural induction over $\phi$:
  \begin{enumerate}
    \item If $\phi=\tt$ or $\phi=\ff$, one can prove easily the lemma holds;
    \item If $\phi=p$ is a literal, then we know $\of{\phi}=p$ and $\olg{\phi}=\{\{p\}\}$. Thus 
    $\of{\phi}\equiv\bigvee_{O\in\olg{\phi}}\bigwedge_{l\in O} l$ is true;
    \item If $\phi=X\psi$, then we know $\of{\phi}=\of{\psi}$ and $\olg{\phi}=\olg{\psi}$. By induction 
    hypothesis we have $\of{\psi}\equiv \bigvee_{O\in\olg{\psi}}\bigwedge_{l\in O} l$ holds. So 
    it is also true that $\of{\phi}\equiv \bigvee_{O\in\olg{\phi}}\bigwedge_{l\in O} l$;
    \item If $\phi=\phi_1U\phi_2$ or $\phi=\phi_1R\phi_2$, then we know $\of{\phi}=\of{\phi_2}$ and 
    $\olg{\phi}=\olg{\phi_2}$. By induction hypothesis we have $\of{\phi_2}\equiv \bigvee_{O\in\olg{\phi_2}}\bigwedge_{l\in O} l$ holds. So 
    it is also true that $\of{\phi}\equiv \bigvee_{O\in\olg{\phi}}\bigwedge_{l\in O} l$;
    \item If $\phi=\phi_1\vee\phi_2$, then we know $\of{\phi}=\of{\phi_1}\vee\of{\phi_2}$ and 
    $\olg{\phi}=\olg{\phi_1}\cup \olg{\phi_2}$. By induction hypothesis we have $\of{\phi_i}\equiv 
    \bigvee_{O_i\in\olg{\phi_i}}\bigwedge_{l\in O_i}l$ holds, where $i=1,2$. Then it is true that 
    $\of{\phi}\equiv 
    \of{\phi_1}\vee\of{\phi_2}\equiv\bigvee_{O_1\in\olg{\phi_1}}\bigwedge_{l\in O_1}l \vee \bigvee_{O_2\in\olg{\phi_2}}\bigwedge_{l\in O_2}l \equiv \bigvee_{O\in\olg{\phi}}\bigwedge_{l\in O}l$;
    \item If $\phi=\phi_1\wedge\phi_2$, then we know $\of{\phi}=\of{\phi_1}\wedge\of{\phi_2}$ and 
    $\olg{\phi}=\{O_1\cup O_2| O_i\in\olg{\phi_i}(i=1,2)\}$.  By induction hypothesis we have $\of{\phi_i}\equiv 
    \bigvee_{O_i\in\olg{\phi_i}}\bigwedge_{l\in O_i}l$ holds, where $i=1,2$. Then it is true that 
    $\of{\phi}\equiv 
    \of{\phi_1}\wedge\of{\phi_2}\equiv\bigvee_{O_1\in\olg{\phi_1}}\bigwedge_{l_1\in O_1}l_1 \wedge \bigvee_{O_2\in\olg{\phi_2}}\bigwedge_{l_2\in O_2}l_2 \equiv \bigvee_{O_1\in\olg{\phi_1}}\bigvee_{O_2\in\olg{\phi_2}}\bigwedge_{l_1\in O_1}\bigwedge_{l_2\in O_2}(l_1\wedge l_2)\equiv\\\bigvee_{O_1\cup O_2\in \olg{\phi}}\bigwedge_{l_1\wedge l_2\in O_1\cup O_2}(l_1\wedge l_2)\equiv \\\bigvee_{O\in\olg{\phi}}\bigwedge_{l\in O}l$. The proof is done.
  \end{enumerate}
\end{proof}

\subsection{Obligation-based Satisfiability Checking Revisited}

In this section, we adapt our general checking theorem (Theorem
\ref{thm:scc}) via reducing checking the containment of an obligation
to the satisfiability of the corresponding obligation formula. Lemma
\ref{lem:reduction_2} below shows the reduction first and Theorem
\ref{thm:satscc} tells how to achieve the general checking via the
obligation formula. Before that, we introduce the $\models_w$ (weak satisfaction relation) 
operator appeared in the theorem.

Let $S \subseteq L$ be a set of literals of $L_\phi$, and $\alpha$ a
propositional formula in NNF. We define $S\models_w \alpha$ in
a syntactic way: if $\alpha$ is a literal, $\tt$ or $\ff$ then $S\models_w \alpha$
iff $\alpha\in S$, $S\models_w \alpha_1\wedge\alpha_2$ iff
$S\models_w\alpha_1$ and $S\models_w\alpha_2$, and
$S\models_w \alpha_1\vee\alpha_2$ iff $S\models_w\alpha_1$ or
$S\models_w\alpha_2$. Note $S$ needs not to be
consistent, e.g., $\{a,\neg a\}\models_w a\wedge \neg a$ holds according to the definition.

\begin{lemma}\label{lem:reduction_2}
  Given an LTL formula $\phi$ and a literal set $S$, then $S\models_w \of{\phi}$ iff there exists 
  an obligation $O\in\olg{\phi}$ such that $O\subseteq S$.
\end{lemma}
\begin{proof}
   According to Lemma \ref{lemma:ofandos}, $\olg{\phi}$ is semantically equivalent to the DNF of $\of{\phi}$. 
   Then from Definition \ref{def:os} we know an obligation in $\olg{\phi}$ is essentially a clause of the DNF of $\of{\phi}$. 
   And it is obvious that $S\models_w \of{\phi}$ iff there is a clause $cl$ in the DNF of $\of{\phi}$ which 
   satisfies $S\models_w \bigwedge cl$, i.e., $cl\subseteq S$. Let $O=cl$
   and we know $O$ is an obligation in $\olg{\phi}$. The proof is done.
   
\end{proof}


\begin{theorem}[SAT-Based Generalized Satisfiability Checking]\label{thm:satscc}
  The LTL formula $\phi$ is satisfiable iff there exists a SCC scc and a state $\psi\in scc$ 
  in $T_\phi$ such that $L(scc)\models_w \of{\psi}$.
\end{theorem}
\begin{proof}
  First according to Lemma \ref{lem:reduction_2} we know $L(scc)\models_w \of{\psi}$ holds iff 
  there exists an obligation $O\in\olg{\phi}$ such that $O\subseteq L(scc)$. Then from  
  Theorem \ref{thm:scc} we can directly conclude this theorem.
\end{proof}

In Theorem \ref{thm:satscc} the set $L(scc)$ collects all literals along $scc$, thus it may 
be inconsistent. So it is necessary to introduce the notation $\models_w$.
  
\section{Satisfiability Checking Acceleration}\label{sec:ofa}

In this section we present accelerating techniques exploiting
obligation formulas that are tailored to both satisfiable and
unsatisfiable formulas.

\subsection{Acceleration on Satisfiable formulas }
Recall that we need to find a consistent obligation in $\olg{\phi}$ in Theorem \ref{thm:oa}. 
Now the problem can be reduced to that of  
checking whether $\of{\phi}$ is satisfiable. The following lemma shows that if $\of{\phi}$ is 
satisfiable then there exists a consistent obligation in $\olg{\phi}$. 

\begin{lemma}\label{lem:reduction}
  For an LTL formula $\phi$, if $\of{\phi}$ is satisfiable, then there exists a consistent 
  obligation $O\in \olg{\phi}$. 
\end{lemma}
\begin{proof}
  According to Lemma \ref{lemma:ofandos}, $\olg{\phi}$ is semantically
  equivalent to the DNF of $\of{\phi}$. So every
  obligation in $\olg{\phi}$ is actually a clause in the DNF of
  $\of{\phi}$. And it is apparently true that $\of{\phi}$ is
  satisfiable implies there exists a clause $cl$ in the DNF of
  $\of{\phi}$ which is satisfiable. Thus $cl$ is consistent,
  i.e. $\bigwedge cl\not=\ff$. Let $O=cl$ and we know that $O$ is an
  consistent obligation in $\olg{\phi}$. The proof is done.
  
\end{proof}

From Lemma \ref{lem:reduction}, Theorem \ref{thm:oa} can be slightly adapted to obtain 
our SAT-based obligation acceleration for satisfiable formulas:

\begin{theorem}[SAT-Based Obligation Acceleration]\label{thm:satoa}
  For an LTL formula $\phi$, if $\of{\phi}$ is satisfiable, then $\phi$ is also satisfiable. 
\end{theorem}

\subsection{Acceleration on Unsatisfiable formulas}
The previous section proposes a heuristic for checking satisfiability
of obligation formulas. In this section we further
exploit SAT solvers to develop heuristics for checking unsatisfiable
formulas by using the obligation formulas. We first use an example to
explain our idea. Consider the formula $\phi=G a \wedge X\neg a$.  One
can see that $\phi$ is unsatisfiable. If we look into the formula, $a$
must be true in every position from the beginning (position $0$) in
$Ga$, on the other side, $a$ must be false in the position $1$ due to
$X\neg a$: this is obviously a contradiction.  Now recall our approach: $\of{\phi} =
a\wedge\neg a$ is unsatisfiable, so Theorem \ref{thm:satoa} cannot
apply. The observation we get here is that there is no
positional information for literals in $\of{\phi}$ so that we lost the
information that $a$ and $\neg a$ must both be true in position $1$.

For this purpose, we extend the \textit{obligation formula} for a
formula $\phi$, denoted as $\ofp{\phi}$, with additional positional
information for each literal. Besides the literal itself, the start
position and its duration are also recorded in $\ofp{\phi}$. We denote
the alphabet of $\ofp{\phi}$ as $\mathcal{L}=L_\phi\times
\mathbb{N}\cup\{\nondeter\}\times\{\cur,\inf,\geq\}$, where
each $l\in
\mathcal{L}$ consists of three elements:
\begin{itemize}
\item the propositional property ($L$),
\item start position ($\mathbb{N}\cup\{\nondeter\}$) from which the property must be
  satisfied. The symbol $\nondeter$ means the start position is not
  determined.
\item its duration ($\{\cur,\inf,\geq\}$) where $\cur$ means the duration
  is just the start position, $\geq$ means the duration is all from
  the start position; and $\inf$ means the duration is
  infinitely many from the start position, but not all. 
\end{itemize}

For convenience in the following, we use the notations $l.prop$,
$l.start$, and $l.duration$ to represent its corresponding first,
second and third elements for $l\in \mathcal{L}$. So, if $l=\langle p,
0, \geq\rangle$, then $l.prop=p$, $l.start=0$ and
$l.duration=\geq$. We also use the notation $[\ofp{\phi}]$ ($[\of{\phi}]$) to represent 
the set of literals appearing in $\ofp{\phi}$ ($\of{\phi}$).  Now we give the
formal definition of $\ofp{\phi}$:

\begin{definition}[Obligation Formula with Position]\label{def:ofp}
    Given an LTL formula $\phi$, the corresponding \textit{obligation formula with position}, denoted 
    as $\ofp{\phi}$, is defined recursively as follows:
    \begin{itemize}
      \item If $\phi=p$: $\ofp{\phi} = \langle p, 0, \cur\rangle$;
      \item If $\phi = \phi_1 \wedge \phi_2$: $\ofp{\phi} = \ofp{\phi_1} \wedge \ofp{\phi_2}$;
      \item If $\phi = \phi_1 \vee \phi_2$: 
      \begin{itemize}
        \item if for every 
      $l_1, l_2\in [\ofp{\phi_1}]\cup [\ofp{\phi_2}]$ it holds that $l_1.start=l_2.start$, then 
      $\ofp{\phi} = \ofp{\phi_1}\vee \ofp{\phi_2}$; 
    \item Otherwise $\ofp{\phi}=\ofp{\phi_1}'\vee \ofp{\phi_2}'$,
      where $\ofp{\phi_i}' (i=1,2)$ is acquired from $\ofp{\phi_i}$ by setting
      $l.start = \nondeter$ and $l.duration=\cur$ for every $l\in
      [\ofp{\phi_i}]$;
      \end{itemize}
      \item If $\phi = X\psi$: $\ofp{\phi} = Pos(\ofp{\psi}, X)$; 
      \item If $\phi = \phi_1 U \phi_2$: $\ofp{\phi} = Pos(\ofp{\phi_2}, U)$;
      \item If $\phi = \phi_1 R \phi_2$ : $\ofp{\phi} = Pos(\ofp{\phi_2}, R)$;       
      \item If $\phi = G\psi$: $\ofp{\phi} = Pos(\ofp{\psi}, G)$;
    \end{itemize}
    where the function $Pos(\ofp{\psi}, type)$ updates $\ofp{\psi}$
    via the $type$.  Explicit rules are listed in
    Table \ref{tab:pos}.
  
\begin{table}
\caption{The explicit rules for the $Pos$ function}\label{tab:pos}
\centering
\scalebox{1}{
    \begin{tabular}{|c|c|c|c|c|}
    \hline
      Literal & X & U & R & G \\
      \hline
      $\langle p, i, \cur\rangle$ & $\langle p, i+1, \cur\rangle$ & $\langle p, \nondeter, \cur\rangle$ & $\langle p, i, \cur\rangle$ & $\langle p, i, \geq\rangle$\\
      \hline
      $\langle p, \nondeter, \cur\rangle$ & $\langle p, \nondeter, \cur\rangle$ & $\langle p, \nondeter, \cur\rangle$ & $\langle p, \nondeter, \cur\rangle$ & $\langle p, \nondeter, \inf\rangle$\\
      \hline
      $\langle p, i, \ge\rangle$ & $\langle p, i+1, \ge\rangle$ & $\langle p, \nondeter, \ge\rangle$ & $\langle p, i, \geq\rangle$ & $\langle p, i, \geq\rangle$\\
      \hline
      $\langle p, \nondeter, \geq\rangle$ & $\langle p, \nondeter, \geq\rangle$ & $\langle p, \nondeter, \geq\rangle$ & $\langle p, \nondeter, \geq\rangle$ & $\langle p, \nondeter, \geq\rangle$\\
      \hline
      $\langle p, i, \inf\rangle$ & $\langle p, i+1, \inf\rangle$ & $\langle p, \nondeter, \inf\rangle$ & $\langle p, i, \inf\rangle$ & $\langle p, i, \inf\rangle$\\
      \hline
      $\langle p, \nondeter, \inf\rangle$ & $\langle p, \nondeter, \inf\rangle$ & $\langle p, 
\nondeter, \inf\rangle$ & $\langle p, 
\nondeter, \inf\rangle$ & $\langle p, 
\nondeter, \inf\rangle$\\
    \hline
    \end{tabular}
}
    
\end{table}
\end{definition}

The $\vee$ operator is a key which causes nondeterminism. So every start position and duration in 
literals should be updated to $
\nondeter$ and $\cur$ respectively -- unless we make sure all literals' start 
positions are the same.
The first column of Table \ref{tab:pos} shows all possible compositions for literals. The second to 
fifth columns show the new composition after the corresponding temporal operator acting on the 
literal. The $X$ operator only add 1 to the start position if it is determined, and the $R$ operator 
does not change the original information at all. For the $U$ operator it makes every start 
position undetermined. The $G$ operator is distinguished with $R$ as it causes the duration. If its 
nested literal $l$ satisfies $l.start=
\nondeter$ or $l.duration=\inf$, then it will update 
$l.duration=\inf$; otherwise it updates $l.duration=\geq$.

It should be mentioned that $\ofp{\phi}$ is essentially an extended propositional formula whose alphabet is 
$\mathcal{L}=L_\phi\times\mathbb{N}\cup\{\nondeter\}\times\{\cur,\inf,\geq\}$. Definition 
\ref{def:ofp} only involves in the syntactic level of $\ofp{\phi}$, and its semantics is skipped as 
we treat $\ofp{\phi}$ an intermediate structure but actually set up the decision procedure on the \textit{positional projection formulas} 
created from $\ofp{\phi}$. The definition is shown below. 

So far we have encoded the positional information into the literals and \textit{obligation formulas}. 
The following definition provides us a mechanism to project the \textit{obligation formula} into each position we 
concern. We try to make the projection loose enough to guarantee the correctness: In the definitions, if 
the literal is not determined in the projecting position, then we just assign 
its projection to be $\tt$.

\begin{definition}[Positional Projection on Obligation Formulas]\label{def:ofp_projection}
    Given an obligation formula with positions $\ofp{\phi}$ from $\phi$, its projection under the 
    position $i$, denoted 
    as $\ofp{\phi}\!\downarrow_i$, is defined recursively as follows:
    \begin{itemize}
      \item If $\ofp{\phi} = l$: \[
                                 \ofp{\phi}\!\downarrow_i=
                                 \begin{cases}
                                 l.prop & \textit{if } l.start = i \text{, or }\\
                                  & l.start < i \text{ and }l.duration=\geq\text{;}\\
                                 \tt & \text{otherwise.}
                                 \end{cases}
                              \]
      \item If $\ofp{\phi} = t_1\wedge t_2$: $\ofp{\phi}\!\downarrow_i = t_1\!\downarrow_i \wedge t_2\!\downarrow_i$;
      \item If $\ofp{\phi} = t_1\vee t_2$: $\ofp{\phi}\!\downarrow_i = t_1\!\downarrow_i \vee t_2\!\downarrow_i$.
    \end{itemize}
\end{definition}

Informally speaking, $\ofp{\phi}\!\downarrow_i$ keeps the first part of literals whose projection on position $i$ is true. For these $l\in [\ofp{\phi}]$, it is either $l.start = i$ holds or $l.start < i $ and $l.duration=\geq$ hold. Otherwise the literals are substituted  
by $\tt$. So $\ofp{\phi}\!\downarrow_i$ is a pure propositional formula. 

For example, 
consider the formula $\phi=GX(a\wedge bUc)$ and thus $\ofp{\phi}=\langle a, 1, \geq\rangle \wedge\langle c, 
\nondeter, \geq\rangle$. 
Let $l_1=\langle a, 1, \geq\rangle$ and $l_2=\langle c, 
\nondeter, \geq\rangle$ and we start from the literals. According to 
Definition \ref{def:ofp_projection} we have $l_1\!\downarrow_0\equiv\tt$, 
$l_1\!\downarrow_1\equiv a$ (since $l_1.start=1$) and $l_1\!\downarrow_i\equiv a$ for every $i> 1$ (since $l_1.start<i$ and 
$l_1.duration=\geq$). Note also $l_2\!\downarrow_i\equiv\tt$ for all $i\geq 0$, and it is because $l_2.start=
\nondeter$ which is 
undetermined so that its projection for every position is $\tt$. 
Thus, recursively we know that $\ofp{\phi}\!\downarrow_0=\tt\wedge\tt=\tt$, $\ofp{\phi}\!\downarrow_1=a\wedge\tt=a$ and etc.

Now the whole framework has been established, and we can conclude the formula $\phi$ is 
unsatisfiable via finding there is a position which cannot be satisfied in all its models: this is exactly what 
Theorem \ref{thm:satunsat} below talks about. Before that, Lemma \ref{lem:reduction_3}  should be introduced at first, which shows the truth of the reverse of Theorem \ref{thm:satunsat}. In the lemma, the notation $\xi(i)$ represents the $i$th 
element of the infinite trace $\xi$.

\begin{lemma}\label{lem:reduction_3}
    Given an infinite word $\xi$ and an LTL formula $\phi$, if $\xi\models\phi$, then 
    for every position $i\ge 0$ it holds that $\xi(i)\models \ofp{\phi}\!\downarrow_i$.
\end{lemma}
\begin{proof}
  We prove this lemma by structural induction over $\phi$. 
  \begin{enumerate}
  
    \item If $\phi=p$ is a literal, from Definition \ref{def:ofp} we know 
    $\ofp{\phi} = \langle p, 0, \cur\rangle$, and from Definition \ref{def:ofp_projection} 
    we know $\ofp{\phi}\!\downarrow_0=p$ and 
    $\ofp{\phi}\!\downarrow_i=\tt (i\geq 1)$. Since $\xi\models \phi=p$ so $\xi(0)\models p$ must 
    hold according to the LTL semantics. For $i\geq 1$ it is obviously true that 
    $\xi(i)\models \ofp{\phi}\!\downarrow_i=\tt$; 
    
    \item If $\phi=X\psi$, since $\xi\models\phi$ so $\xi_1\models\psi$ according to the 
    LTL semantics. By induction hypothesis, we know that $\xi_1 (i)\models \ofp{\psi}\!\downarrow_i$ 
    for every $i\geq 0$. Then according to the $Pos$ rules on $X$ operator in Table \ref{tab:pos}, we 
    know that $\xi\models \ofp{\phi}\!\downarrow_i$ for every $i\geq 1$. Also we know 
    $\ofp{\phi}\!\downarrow_0=\tt$, so it is apparent that $\xi(0)\models \ofp{\phi}\!\downarrow_0=\tt$;
    
    \item If $\phi=\phi_1 U \phi_2$, then according to the $Pos$ rule on $U$ operator in Table 
    \ref{tab:pos} we know that $l.start=
\nondeter$ for all $l\in [\ofp{\phi}]$. Thus from Definition 
    \ref{def:ofp_projection}, $ofp(\phi)\!\downarrow_i=\tt$ holds 
    for every $i\geq 0$. Hence $\xi(i)\models \ofp{\phi}\!\downarrow_i=\tt$ holds for all $i\geq 0$;
    
    \item If $\phi=\phi_1 R \phi_2$, then first we know $\xi\models\phi$ implies 
    $\xi\models\phi_2$ from the LTL semantics on R operator. Thus by induction hypothesis we already 
    have $\xi(i)\models \ofp{\phi_2}\!\downarrow_i$ for every $i\geq 0$. Moreover it is true that 
    $\ofp{\phi_2}=\ofp{\phi}$ according to the $Pos$ rules on $R$ operator in Table \ref{tab:pos}, so 
    does $\ofp{\phi_2}\!\downarrow_i=\ofp{\phi}\!\downarrow_i$ for all $i\geq 0$. Thus it concludes that 
    $\xi(i)\models \ofp{\phi}\!\downarrow_i$ for every $i\geq 0$;

    \item If $\phi=G\psi$, then $\xi\models\phi$ implies that $\xi_i\models\psi$ for all $i\geq 0$. 
    By induction hypothesis, for every $i\geq 0$ and $j\geq 0$ we have the assumption that 
    $\xi_j(i)\models \ofp{\psi}\!\downarrow_i$. Now we consider the possibilities of $\ofp{\phi}$:
    \begin{itemize}
      \item If $\ofp{\psi}=\langle p, n, \cur\rangle$, then from Definition \ref{def:ofp} 
      we know $\ofp{\phi}=Pos(\ofp{\psi}, G)=\langle p, n, \geq\rangle$. For $0\leq k < n$ we have $\xi(k)\models\ofp{\phi}\!\downarrow_k=\tt$ according to Definition \ref{def:ofp_projection}. And for $k\geq n$ we know 
      that $\ofp{\phi}\!\downarrow_k=p$. Since $\xi_j(i)\models \ofp{\psi}\!\downarrow_i$ for every $i,j\geq 0$ and $\ofp{\psi}=\langle p, n, \cur\rangle$, so $\xi(k)\models p$ also holds for $k\geq n$. 
      Thus for every $k\geq 0$ we have $\xi(k)\models\ofp{\phi}\!\downarrow_k$;
      \item If $\ofp{\psi}=\langle p, n, \inf\rangle$, then from Definition \ref{def:ofp} 
      we know $\ofp{\phi}=Pos(\ofp{\psi}, G)=\langle p, n, \inf\rangle$. For $0\leq k < n$ we have $\xi(k)\models\ofp{\phi}\!\downarrow_k=\tt$. And $\xi(k)\models\ofp{\phi}\!\downarrow_k=p$ holds if $k=n$. For $k> n$ we know $\ofp{\phi}\!\downarrow_k=\tt$ so $\xi(k)\models\ofp{\phi}\!\downarrow_k$ is always true;
      \item If $\ofp{\psi}=\langle p, \nondeter, -\rangle$ where $-$ can be $\cur$, $\inf$ or $\geq$, then 
      from Definition \ref{def:ofp} we have $\ofp{\phi}=Pos(\ofp{\psi}, G)=\langle p, \nondeter, -\rangle$. Thus according to Definition \ref{def:ofp_projection} we know $\ofp{\phi}\!\downarrow_k=\tt$ for every 
      $k\geq 0$. So $\xi(k)\models\ofp{\phi}\!\downarrow_k$ is always true;
      \item Inductively if $\ofp{\psi}=t_1\wedge t_2$, then we know $\ofp{\psi}\!\downarrow_i=t_1\!\downarrow_i\wedge t_2\!\downarrow_i$ and $\xi_j(i)\models \ofp{\psi}\!\downarrow_i$ implies $\xi_j(i)\models t_1\!\downarrow_i$ and $\xi_j(i)\models t_2\!\downarrow_i$ hold for every $i,j\geq 0$. By inductive hypothesis we have proven that $\xi(k)\models Pos(t_1, G)\!\downarrow_k$ and $\xi(k)\models Pos(t_2, G)\!\downarrow_k$ for $k\geq 0$, so $\xi(k)\models Pos(t_1, G)\!\downarrow_k\wedge Pos(t_2, G)\!\downarrow_k$ also holds. As we know $\ofp{\phi}=Pos(t_1,G)\wedge Pos(t_2,G)$, so it is true that 
      $\xi(k)\models\ofp{\phi}\!\downarrow_k$ for $k\geq 0$;
      \item If $\ofp{\psi}=t_1\vee t_2$, then we know $\ofp{\psi}\!\downarrow_i=t_1\!\downarrow_i\vee t_2\!\downarrow_i$ and $\xi_j(i)\models \ofp{\psi}\!\downarrow_i$ implies $\xi_j(i)\models t_1\!\downarrow_i$ or $\xi_j(i)\models t_2\!\downarrow_i$ holds for every $i,j\geq 0$. By inductive hypothesis we have proven that $\xi(k)\models Pos(t_1, G)\!\downarrow_k$ or $\xi(k)\models Pos(t_2, G)\!\downarrow_k$ for $k\geq 0$, so $\xi(k)\models Pos(t_1, G)\!\downarrow_k\vee Pos(t_2, G)\!\downarrow_k$ also holds. 
      According to Definition \ref{def:ofp} if $\ofp{\phi}=Pos(t_1,G)\vee Pos(t_2,G)$ then it is true that $\xi(k)\models\ofp{\phi}\!\downarrow_k$ for $k\geq 0$. And if $\ofp{\phi}=Pos(t_1,G)'\vee Pos(t_2, G)'$ from Definition \ref{def:ofp} then we know $\ofp{\phi}\!\downarrow_k=\tt$ for every $k\geq 0$, so $\xi(k)\models\ofp{\phi}\!\downarrow_k$ holds as well. 
    \end{itemize}
    Thus, we prove that $\xi(i)\models\ofp{G\psi}\!\downarrow_i$ holds for every $i\geq 0$;

    \item If $\phi=\phi_1\wedge \phi_2$, then from Definition \ref{def:ofp} we know that 
    $\ofp{\phi}=\ofp{\phi_1}\wedge \ofp{\phi_2}$, and so does 
    $\ofp{\phi}\!\downarrow_i=\ofp{\phi_1}\!\downarrow_i\wedge \ofp{\phi_2}\!\downarrow_i$ 
    for $i \geq 0$ via Definition \ref{def:ofp_projection}. 
    Also $\xi\models\phi=\phi_1\wedge\phi_2$ implies $\xi\models\phi_1$ and $\xi\models\phi_2$. 
    By induction hypothesis, it holds that $\xi(i)\models \ofp{\phi_1}\!\downarrow_i$ and 
    $\xi(i)\models \ofp{\phi_2}\!\downarrow_i$ for every $i\geq 0$. 
    So $\xi(i)\models \ofp{\phi_1}\!\downarrow_i\wedge \ofp{\phi_2}\!\downarrow_i$, 
    which proves that $\xi(i)\models \ofp{\phi_1\wedge\phi_2}\!\downarrow_i$ for $i\geq 0$;

    \item If $\phi=\phi_1\vee\phi_2$, then $\xi\models\phi$ implies either $\xi\models\phi_1$ or 
    $\xi\models\phi_2$ holds. Assume that $\xi\models\phi_1$ holds. 
    Also according to Definition \ref{def:ofp}, there are two possibilities 
    on $\ofp{\phi}$: 1) If $\ofp{\phi}= \ofp{\phi_1}\vee \ofp{\phi_2}$, then by induction hypothesis 
    we know that $\xi\models\phi_1$ implies $\xi(i)\models \ofp{\phi_1}\!\downarrow_i$ for all 
    $i\geq 0$. Moreover, we can conclude that 
    $\ofp{\phi_1}\!\downarrow_i\Rightarrow \ofp{\phi}\!\downarrow_i$ from Definition 
    \ref{def:ofp_projection}. Combining the conclusion above we can 
    finally prove $\xi(i)\models \ofp{\phi}\!\downarrow_i$ for all 
    $i\geq 0$; 2) If $\ofp{\phi}= \ofp{\phi_1}'\vee \ofp{\phi_2}'$, then since $l.start$ is updated to 
    $
\nondeter$ for every $l$ in 
    $[\ofp{\phi_1}']$ and $[\ofp{\phi_2}']$: it causes that $\ofp{\phi_1}'\!\downarrow_i$ and 
    $\ofp{\phi_2}'\!\downarrow_i$ are assigned to $\tt$ according to Definition 
    \ref{def:ofp_projection}, and so does $\ofp{\phi}$. Hence it is 
    easy to check that $\xi(i)\models \ofp{\phi}\!\downarrow_i=\tt$ for all 
    $i\geq 0$. Finally the proof is done.
  \end{enumerate}
  
\end{proof}

Lemma \ref{lem:reduction_3} directly implies the following theorem for checking unsatisfiable 
formulas:  

\begin{theorem}[SAT-Based Unsatisfiable Checking]\label{thm:satunsat}
    Given an LTL formula $\phi$, if there exists a position $i\ge 0$ such that 
    $\ofp{\phi}\downarrow_i$ 
    is unsatisfiable, then $\phi$ is also unsatisfiable. 
\end{theorem}

However, this theorem can only be implemented as a heuristics technique because we cannot check 
every position of an infinite model in the worst case. On the other hand, it is also not necessary 
to check the 
accurate position every time: instead we can find the unsatisfiable position in a more 
abstract way. 
To achieve this, we need to introduce a more abstract definition for projection on obligation formulas.

\begin{definition}[Abstract Projection on Obligation Formulas]\label{def:ofp_abstract}
  Given an obligation formula with positions $\ofp{\phi}$ from $\phi$
  and a literal set $S$, we define its projection under $S$, denoted
  as $\ofp{\phi}\downarrow_S$, as follows:
    \begin{itemize}
      \item $\ofp{\phi} = l$: if $l\in S$ then $\ofp{\phi}\downarrow_S=l.prop$, else 
      $\ofp{\phi}\downarrow_S=\tt$;
      \item $\ofp{\phi} = t_1\wedge t_2$: $\ofp{\phi}\downarrow_S = t_1\downarrow_S \wedge t_2\downarrow_S$;
      \item $\ofp{\phi} = t_1\vee t_2$: $\ofp{\phi}\downarrow_S = t_1\downarrow_S \vee t_2\downarrow_S$;
    \end{itemize}

\end{definition}

Informally speaking, $\ofp{\phi}\downarrow_S$ is a Boolean formula in which literals not in $S$ are 
replaced by $\tt$, and those in $S$ are replaced by their first elements. 
The following corollary lists the strategies we apply in our algorithm. 

\begin{corollary}\label{coro:abstract}
  Given an LTL formula $\phi$, let $S=\{l \mid l\in [\ofp{\phi}] \wedge
  l.duration=\geq\}$, then $\phi$ is unsatisfiable if one of the
  following conditions is true:
  \begin{enumerate}
    \item There exists $l\in [\ofp{\phi}]$ such that $l.start=i$, and $\ofp{\phi}\downarrow_i\equiv\ff$;
    \item $\ofp{\phi}\downarrow_S\equiv\ff$;
    \item There exists $l\in [\ofp{\phi}]$ such that $l.start=
      \nondeter$, and $\ofp{\phi}\downarrow_{S'\cup\{l\}}\equiv\ff$,
      where $S'\subseteq S$ and $l.start = 0$ for each $l\in S'$;
    \item There exists $l\in [\ofp{\phi}]$ such that $l.duration=\inf$, and 
    $\ofp{\phi}\downarrow_{S\cup\{l\}}\equiv\ff$.
  \end{enumerate}
\end{corollary}

Note here $S$ cannot be empty. The correctness of the corollary is
guaranteed by Theorem \ref{thm:satunsat}: An unsatisfiable position
can always be found in above four conditions. Since the number of
literals is linear to the size of $\phi$, so the additional cost for
the unsatisfiable checking is polynomial to the size of $\phi$. Below
we use several examples to demonstrate the efficiency of our approach:

\begin{example}
  \begin{enumerate}
  \item Consider the formula $\phi=a\wedge (b R \neg a)$. We have
    $\ofp{\phi}=\langle a, 0, \cur\rangle\wedge\langle \neg a, 0,
    \cur\rangle$. According to the first item of Corollary
    \ref{coro:abstract} we know $\ofp{\phi}\downarrow_0=a\wedge\neg
    a=\ff$. So $\phi$ is unsatisfiable;
  \item Consider the formula $\phi= Ga \wedge G(\neg a\wedge b)$. We
    have $\ofp{\phi}=\langle a, 0, \geq\rangle\wedge\langle \neg a,
    0, \geq\rangle\wedge\langle b, 0, \geq\rangle$. Then from the
    second item of Corollary \ref{coro:abstract} we know $\phi$ is
    unsatisfiable;
    \item For the formula $\phi=F a\wedge G\neg a$, we can use the third item of Corollary 
    \ref{coro:abstract} to check it is unsatisfiable;
    \item For the formula $\phi=Ga \wedge GF\neg a$, the fourth item of Corollary 
    \ref{coro:abstract} can be used to check it is unsatisfiable.
  \end{enumerate}
\end{example}

Note $\ofp{\phi}$ is treated as a proposition formula with the
extended alphabet $\mathcal{L}$. Each element $l$ in $\mathcal{L}$ is
a triple, to keep the positional information. The projections
$\ofp{\phi}\downarrow_i$ and $\ofp{\phi}\downarrow_S$ are
propositional formulas over the literals $L_\phi$, and they are the real ones used
for checking satisfiability in our algorithms.

\section{Experiments}\label{sec:exp}

In this section we first introduce the experimental approach, and then 
present the results.  

\begin{table*}
\renewcommand{\arraystretch}{1.3} 
\caption{Comparison results for the \textbf{Schuppan-collected} benchmarks.
Each cell in the table lists the checking time on satisfiable (above) and 
unsatisfiable (below) cases in the given pattern. For some patterns, the 
results on unsatisfiable formulas are empty as no unsatisfiable cases are 
in these patterns.}\label{tab:schuppan-collected}
\centering
\scalebox{1.2}{
\begin{tabular}{|l|l|l|l|l|l|l|l|l|}
  \hline
  Formula Type  &  pltl  &  ls4  &  TRP++  &  NuSMV-BDD  &  Aalta\_v0.1  &  Aalta\_v0.2\\
\hline
acacia/demo-v3  &  366.8  & 10.9  & \cellcolor{blue!35} 5.9  &  2753.5  &  1765.4  &  495.8\\

\hline
\multirow{2}{*}{alaska/lift}  &  3615.5  &  2155.4  &  10959.3  &  7214.3  &  5828.8  &  \cellcolor{blue!35}1483.4\\

 & 1264.5 & 302.4 & 1604.1 & \cellcolor{green!35} 44.7 & 1521.8 & 1520.1\\
\hline

anzu/amba  & \cellcolor{blue!35} 605.3  &  3913.8  &  4787.6  &  5524.4  &  3421.1  &  1958.9\\

\hline
anzu/genbuf  &  2849.7  &  3725.7  &  5101.4  &  5759.3  &  6207.7  & \cellcolor{blue!35} 1742.9\\
\hline

rozier/counter  &  \cellcolor{blue!35}1415.3  &  2735.1  &  1570.3  &  1502.1  &  3378.9  &  3562.5\\

\hline
\multirow{2}{*}{rozier/formulas}  &  178.5  & \cellcolor{blue!35} 65.3  &  41922.5  &  2433.5  &  2308.7 &  137.3\\

  & 185.8 & \cellcolor{green!35}1.2 & 312.1 &  3.8 & 124.0 & 93.3\\
\hline

rozier/pattern  &  15.1  &  9105.1  &  9897.3  &  9408.6  &  31.484  &  \cellcolor{blue!35} 14.3\\

\hline
\multirow{2}{*}{schuppan/O1formula}  & \cellcolor{blue!35} 0.3  &  629.8  &  461.9  &  505.8  &  517.1  &  1.5\\

  & 1026.9 & 646.1 & 461.3 & 499.8 & 12.9 & \cellcolor{green!35} 1.1\\
\hline
\multirow{2}{*}{schuppan/O2formula}  & \cellcolor{blue!35} 2.7  &   748.0  &  547.7  &  747.1  &  632.7  &   11.1\\

  & 1081.5 & 1081.4 & 1330.9 & 913.7 & 1351.2 & \cellcolor{green!35} 1.2\\
\hline

\multirow{2}{*}{schuppan/phltl}  &  120.5  &  622.1  &  625.5  &  668.7  &  121.0  & \cellcolor{blue!35} 7.1 \\
 & 300.4  &  455.5  &  241.7  & \cellcolor{green!35} 23.4  &  480.0  &  364.8\\
\hline

\multirow{2}{*}{trp/N5x}  & \cellcolor{blue!35} 2.7  &  41.9  &  127.8  &  5.0  &  1556.7  &  11.3\\

  & 21.7 & \cellcolor{green!35}14.4 & 25.3 &  42.4 & 3669.0 & 1345.3\\
\hline
\multirow{2}{*}{trp/N5y}  & \cellcolor{blue!35} 1.1  &  30.6  &  67.9  &  94.3  &  5715.1  &  6.1\\

  & 2760.3 &   380.9 & \cellcolor{green!35} 16.2 &  18.5 & 6760.6 & 2760.4\\
\hline
\multirow{2}{*}{trp/N12x}  &  10095.9  &  431.1  &  1022.1  &  13326.3  &  14836.4  & \cellcolor{blue!35} 36.2\\

  & 9816.5 & \cellcolor{green!35}334.4 & 570.6 & 9465.0 & 9487.2 &  1623.4\\
\hline
trp/N12y  &  106.5  &  591.6  &  1502.1  &  11547.1  &  8677.8  & \cellcolor{blue!35} 12.4\\

  & 4020.5  &  294.3  &  565.8  & \cellcolor{green!35} 74.9  &  4020.9  &  2020.6\\

\hline
\multirow{2}{*}{Total}  &  19375.9  &  24806.4  &  78288.4  &  61069.8  &  76398.6  & \cellcolor{red!35} 9480.8\\

& 20477.8 &  3555.6 & 5217.0 & 11090.2 & 28428.9 & \cellcolor{red!35} 9730.6\\
\hline

\end{tabular}
}

\end{table*}

\subsection{Experimental Strategies}

We use Rice University's SUG@R cluster%
\footnote{\url{http://www.rcsg.rice.edu/sharecore/sugar/}} 
as the experimental platform. The cluster contains 134 Sun Microsystems 
SunFire x4150 nodes, each of which includes 8 cores of 2.83GHz Intel Xeon 
Harpertown CPUs with 16GB RAM. In our experiments, we use Polsat as the 
testing platform \cite{LPZVHCoRR13} to compare the tool with other LTL 
satisfiability solvers. Polsat is run on a node of SUG@R, and the tested 
tools, which are integrated into Polsat, occupy each a unique core (the 
number of tools is less than 8). Timeout limit was set to 60 seconds. 

There are two \Aalta\ versions in our experiments: \Aalta\_v0.1 for
the old tool \cite{LZPVH13} and \Aalta\_v0.2 the current one, which implements 
the algorithms described in this paper%
\footnote{\url{http://www.lab205.org/aalta/}}.  
\Aalta\_v0.2 uses MiniSat \cite{ES03} solver as the SAT solver, and the 
bool2cnf%
\footnote{\url{http://www-ise4.ist.osaka-u.ac.jp/~t-tutiya/sources/bool2cnf/}}
tool to provide the ``DIMACS CNF'' input format for MiniSat. 
Several other LTL satisfiability solvers are also involved in the
experiments. Among them, the pltl tool%
\footnote{Three tool versions can be found at 
\url{http://users.cecs.anu.edu.au/~rpg/PLTLProvers/}; we use the first one, 
following \cite{SD11}.} 
\cite{Sch98} is the representative of tableau-based approach; 
the TRP++ tool \cite{HK03} implements a temporal-resolution strategy; 
the NuSMV tool \cite{CCGGPRST02} uses a model-checking-based method. 
Since NuSMV applies BDD technique, we use NuSMV-BDD to denote it. 
We did not use here NuSMV-BMC (bounded model checking), which can be used 
to check satisfiability, but fails to check unsatisfiability.
We also did not use here the Alaska tool~\cite{DDMR08}, which is the
implementation of antichain-based checking, as it fails to run on SUG@R. 
We did use the fairly recent ls4 tool, which is a SAT-based PLTL 
prover\cite{SW12}%
\footnote{\url{http://www.mpi-inf.mpg.de/~suda/ls4.html}}. 
The input of ls4 is the same as that of TRP++. 
Since the input of TRP++ must be in SNF (Separated Normal Form 
\cite{Fish97}), a SNF generator is also required. 
The \textit{translate} generator is available from the TRP++ website%
\footnote{\url{http://cgi.csc.liv.ac.uk/~konev/software/trp++/}}. 

The tools introduced above use several parameters and their performance may
vary on their selection.  In \cite{SD11}, Schuppan and Darmawan collected 
the winning configurations for each tool.  Some of the parameters in these 
configurations are not, however, available in the tools' new versions; 
instead we chose the parameters that are the closest to the ones in the 
winning configurations. Specifically, for pltl we used the ``-tree'' 
parameter;  for TRP++ we used ``-sBFS -FSR''; for the SNF generator 
\textit{translate} we used ``-s -r''; for NuSMV-BDD we used 
``-dcx -f -flt -dynamic''; and for ls4 we used ``-r2l''. In our experiments,
pltl, TRP++, NuSMV and ls4 versions are r1424, 2.2, 2.5.4 and 1.0.1, 
respectively. 

In the experiments we consider all benchmarks from \cite{RV10,SD11,LZPVH13}.
For simplicity, we call the formulas from \cite{SD11} as 
\textit{schuppan-collected} benchmarks in the following. This benchmark 
contains a total amount of 7446 formulas. To test the scalability of the 
tools on random formulas, we followed \cite{RV10} with length varying from 
100 to 200 and variables number set to 3. For each length we tested a group
of 500 formulas.  We also tested the \textit{random conjunction formulas} 
introduced in \cite{LZPVH13}. A random conjunction has the form of 
$\bigwedge_{1\leq i\leq n} P_i$, where each $P_i$ is a 
\emph{specification pattern} randomly chosen from \cite{DAC98}%
\footnote{\url{http://patterns.projects.cis.ksu.edu/documentation/patterns/ltl.shtml}},
where $n$ here varies from 1 to 20 and for each $n$ we created 500 cases. 
In summary, we tested the tools on approximately 50 patterns and 27,446 
formulas. We do not find any inconsistency among the results from different
tools.

\subsection{Experimental Results}

Table \ref{tab:schuppan-collected} shows the evaluation results on 
formulas from the \textit{schuppan-collected} benchmark.
The first column lists the pattern types, and the second to seventh columns
show the checking time (seconds) of various solvers on each pattern, 
respectively.  Normally, a test benchmark of the given pattern involves two
types of formulas: satisfiable and unsatisfiable ones. We separate them in 
two rows for each entry of benchmark pattern: the upper one is the
checking time for satisfiable formulas and the lower one for
unsatisfiable formulas. Some patterns, however, do not have
unsatisfiable formulas in their data set; we just keep one row in the
table for these patterns, such as \textsf{acacia/demo-v3} and
\textsf{anzu/amba} patterns. We ignored those patterns that cannot 
be decided within the timeout by all solvers.  We also ignored those 
patterns that can be solved within 10 second by all solvers.
We highlight the entry of the best checking result for each pattern: 
blue for satisfiable cases and green for unsatisfiable cases. 

\begin{figure}
\centering
\includegraphics[scale = 0.8]{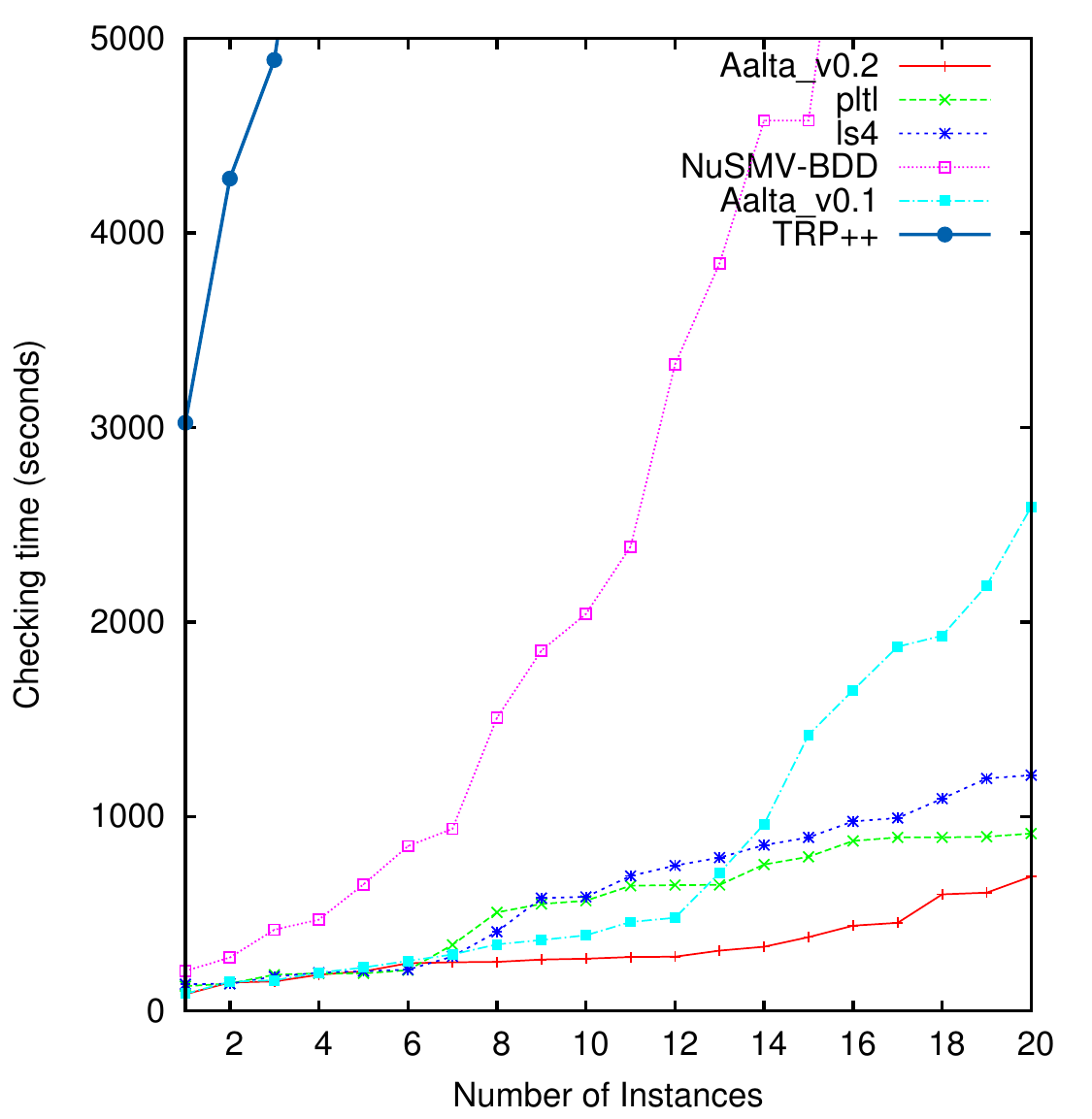}
\caption{Experimental results on extended random formulas with 3 variables.}
\label{fig:random}
\end{figure}

\begin{figure}
\centering
\includegraphics[scale = 0.8]{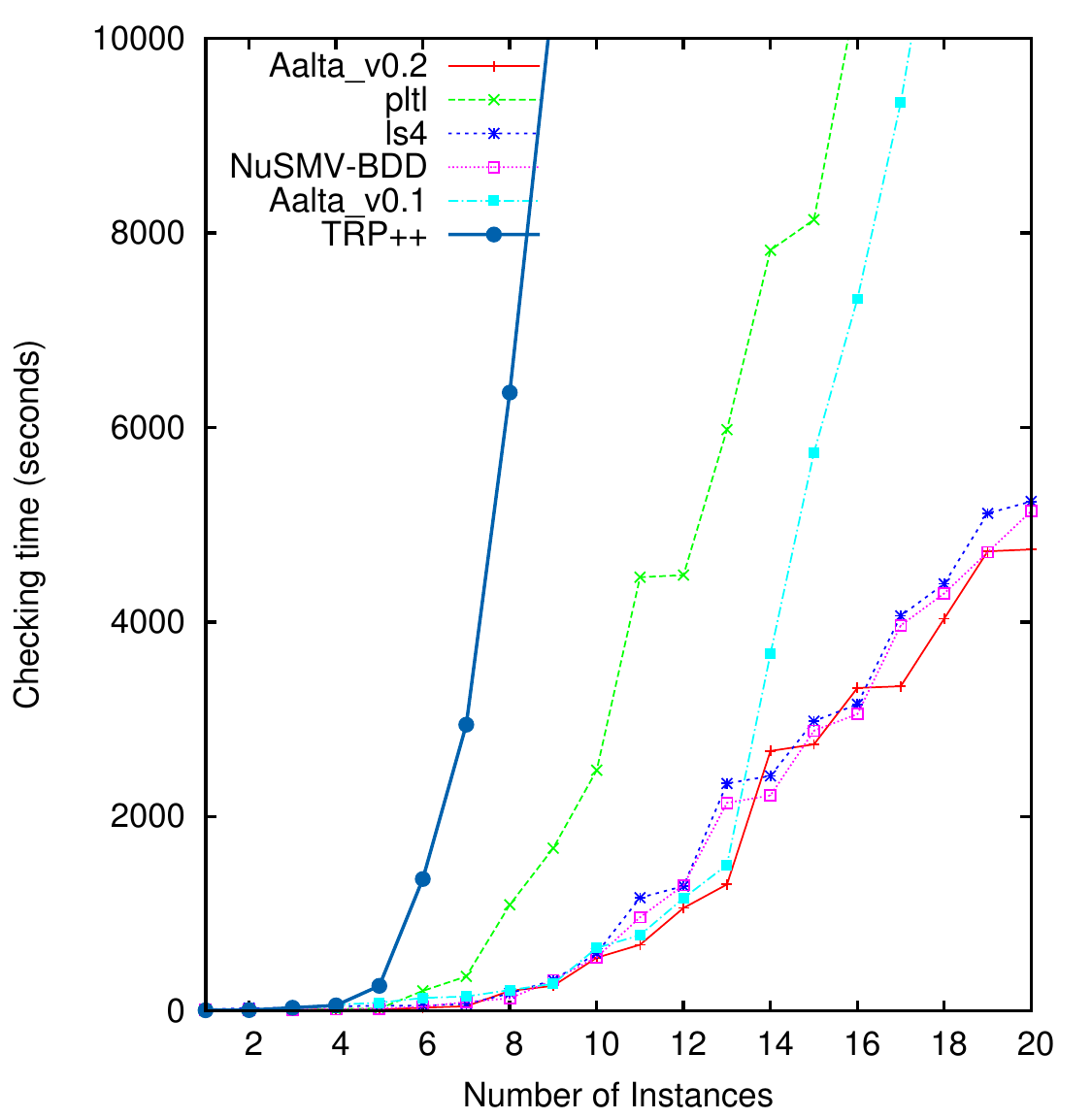}
\caption{Experimental results on random conjunction formulas.}
\label{fig:random_conjunction}
\end{figure}

%

From Table~\ref{tab:schuppan-collected}, we can see that the  proposed 
SAT-based approach dramatically improves the performance of the 
obligation-based satisfiability checking method of previous 
work~\cite{LZPVH13}. For instance, compared to \Aalta\_v0.1, \Aalta\_v0.2 
has a nearly 1000X speedup for the satisfiable cases in 
\textsf{trp/N5y}, and nearly 500X speedup for the unsatisfiable cases in
\textsf{schuppan/O2formula}. In total, \Aalta\_v0.2 performs about 5 times 
better than \Aalta\_v0.1.

Although the experiments confirms the fact that none of investigated solvers
dominates across all patterns, the new tool (\Aalta\_v0.2) has the most 
``wins'' with 8 best results, while pltl has 6 best results, ls4 has 4 best
results, NuSMV-BDD has 3 and TRP++ has another 2 best results. 
Moreover, \Aalta\_v0.2 has the best total time performance,
with the ls4 tool in the second place.

To test the scalability of solvers, we use large random formulas, 
including both \textit{random formulas} and 
\textit{random conjunction formulas}.
The results are shown in Fig. \ref{fig:random} and
Fig. \ref{fig:random_conjunction}, respectively. (We show cactus plots%
\footnote{\url{http://en.wikipedia.org/wiki/Cactus_graph}}, 
where the $x$-axis sorts the instances (a suite of 500 formulas) by 
hardness rather than by length.
Here we do not separate satisfiable and unsatisfiable formulas.

In Fig. \ref{fig:random} one can see the significant improvement 
by the new SAT-based checking framework. (See the gap between the results 
from \Aalta\_v0.1 and \Aalta\_v0.2.). It also see that \Aalta\_v0.2 has 
the best performance for large random formulas. 
In Fig. \ref{fig:random_conjunction} we see that the new tool is almost
best for random conjunction formulas; ls4 and NuSMV-BDD can be competitive
with \Aalta\_v0.2. 

As mentioned earlier, Polsat is not only a testing platform for LTL 
satisfiability solvers, but also a portfolio solver that provides the best 
result by integrating several solvers. We want like to know  how much the 
performance of Polsat is improved via integrating \Aalta\_v0.2 to replace 
\Aalta\_v0.1.  The answer is available in 
Table \ref{tab:schuppan-collected}. For those cases where \Aalta\_v0.2 is
best, there can be a 10X speedup in Polsat with \Aalta\_v0.2;
see \textsf{schuppan/O1formula} and \textsf{schuppan/O2formula} patterns.
\ref{fig:polsat_unsat} show the performance comparison between Polsat with 
\Aalta\_v0.1/\Aalta\_v0.2 on extended satisfiable/unsatisfiable random 
formulas (with lengths from 100 to 200). 
The plots show that the new proposed SAT-based approach boosts the 
performance of Polsat  by 30\% to 60\% on average for random formulas. 

\begin{figure}
\centering
\includegraphics[scale = 0.75]{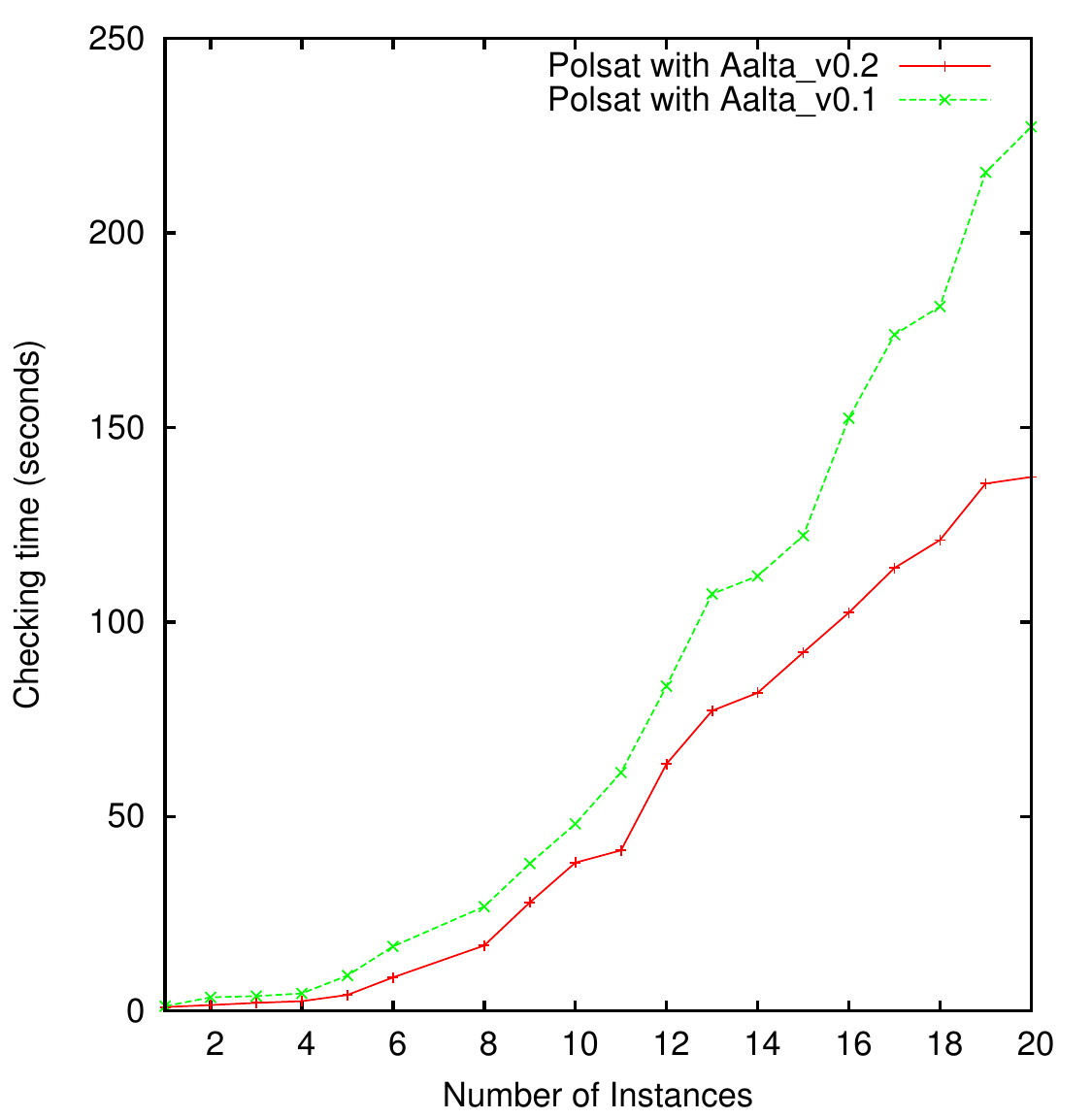}
\caption{Polsat improvement on extended satisfiable random formulas with 3 variables.}
\label{fig:polsat_sat}
\end{figure}

\begin{figure}
\centering
\includegraphics[scale = 0.75]{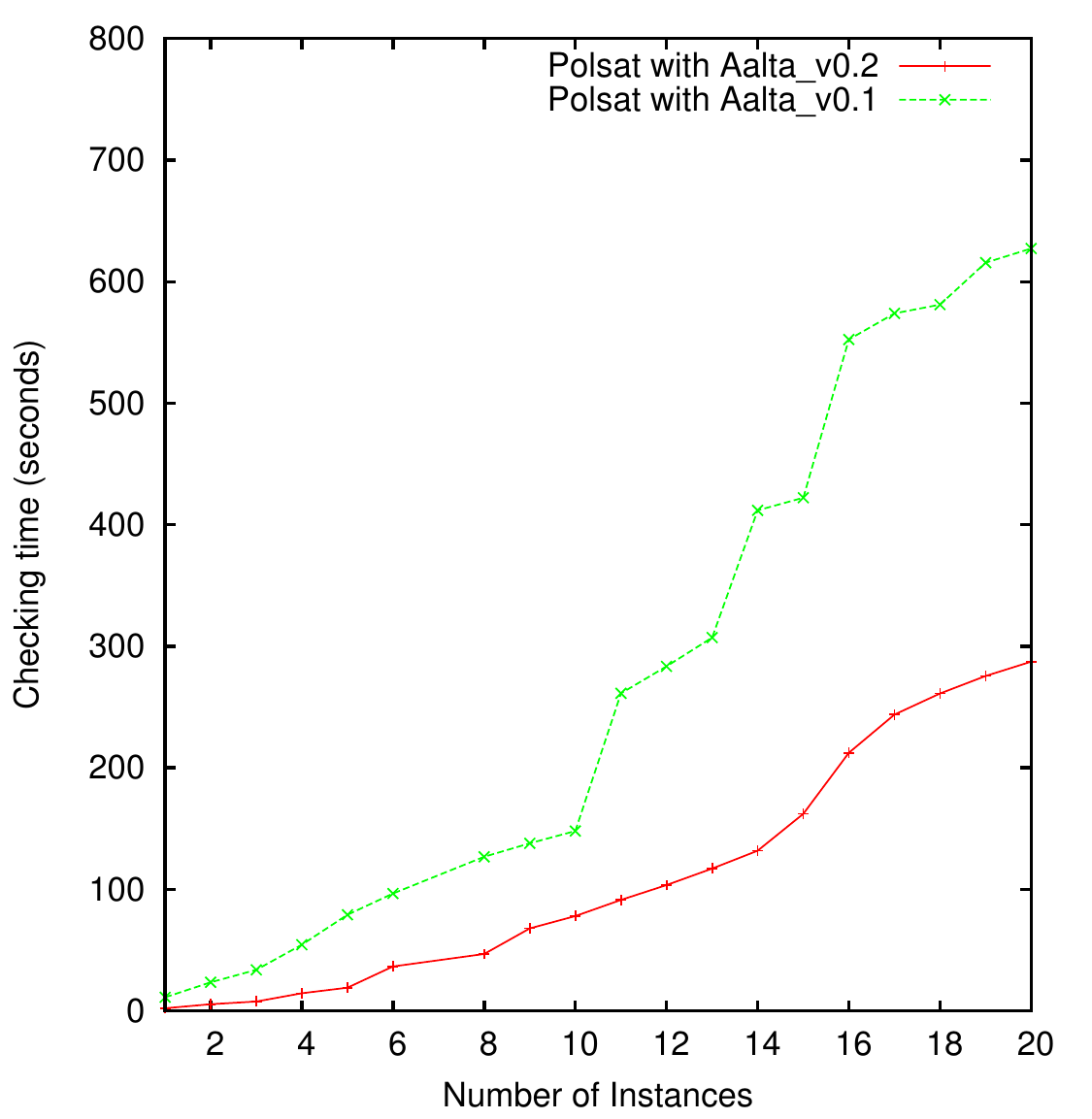}
\caption{Polsat improvement on extended unsatisfiable random formulas with 3 variables.}
\label{fig:polsat_unsat}
\end{figure}

\section{Discussion and Related Work}\label{sec:related} 

The satisfiability-checking framework proposed in this paper is based on LTL 
\textit{normal forms}, which follows the traditional tableau-based expanding 
rules \cite{GPVW95}. By iteratively utilizing the normal-form generation we can 
obtain a \textit{transition system} for a given formula. 
Based on that, a decision procedure leveraging the power of modern SAT solvers is 
achieved, in which the \textit{obligation formula} plays the crucial role. 
It should be mentioned, however that a \emph{tagging} process is required to 
complete this approach.  For example consider $\phi=(a \vee b) U (Ga)$, in which 
the atom $a$ appears twice.  Without tagging, we can see that there exists a transition
$\phi\tran{a}\phi\tran{b}\phi$, which forms a SCC $B$, and
$L(B)=\{a,b\}$ is a superset of the obligation $\{a\}$. 
But, obviously, the infinite path through this SCC can not satisfy
$\phi$. This problem can be solved by tagging the input formula. To simplify 
the presentation here, we omit the details in this paper and refer readers to 
our previous work \cite{LZPVH13}.
  
Several similar ideas about the normal form and transition system are also used 
in other works.  For example, Gerth et al. pointed out every LTL formula has an 
equivalent form of $\bigvee (\alpha_i\wedge X\psi_i)$, when they proved the correctness 
of their LTL-to-GBA (Generalized B\"uchi Automata) translation \cite{GPVW95}. 
Also Duan et al. proposed a similar normal form and normal form graph for PPTL, 
which is considered as a superset of LTL\cite{DTZ08}.  Based on these two concepts, 
they presented a decision procedure for PPTL that can be used in both model checking 
and satisfiability checking. Compared to their work, our contribution 
is the obligation-formula-based checking and the heuristics that can utilize SAT solving. 
We also provide a detailed performance evaluation, which is not included in \cite{DTZ08}.

Bounded model checking (BMC) is the first approach that reduces model checking into SAT 
framework \cite{CBRZ01,CPRS02}. It encodes both the system and property symbolically, and 
unrolls the system step by step to check whether the property is satisfied. As one can keep
unrolling, this approach does not terminate, and one has to select an upper bound
$k$ for unrolling to obtain termination. In the context of LTL satisfiability checking,
BMC provides a method for detecting satisfiability, but not for detecting unsatisfiability.
To overcome this shortage of BMC, full SAT-based model-checking techniques need to be explored. 

Interpolation model checking (IMC) \cite{McM03} extends BMC by separating the $k$-reachability formula
steps into two parts, and generates an interpolant from these sets if the formula is unsatisfiable. 
Here the bound $k$ also needs to be repeatedly incremented, but, unlike BMC, this method is
guaranteed to terminate \cite{McM03}. Other SAT-based model checking techniques are evaluated
in \cite{ADKKM05}, where it is shown that no method dominates across all benchmarks.
One way to obtain a complete SAT-based model checker is by finding the ``threshold'' $ct$ such 
that $M\models_{ct}\phi\Rightarrow M\models\phi$ holds: here $M$ is the system model and $\phi$ is 
a property \cite{CKOS05}. This threshold, however, can be exponential in the size of the system
and the property, and this method does not perform well in practice.

In the experimental part we compare our tool with ls4, which is considered as a SAT-based PLTL prover \cite{SW12}. 
The approach behind the tool does not directly work on PLTL but requires the SNF (Separated Normal Form \cite{Fish97}) input. The 
interesting observation in this strategy is that, the BMC process can become complete by facilitating the SNF's structural 
features. It then takes such a complete SAT-based checking as a ``block'' and successfully extends its application from 
reachability cases to general ones. Our experiments above confirm that this approach performs well on unsatisfiability 
checking. 

Recently IC3 \cite{Bra11} emerged as a popular full SAT-based model-checking technique. 
This approach does not need to unroll the system more than one step, and it keeps an 
inductive invariant inductively during the checking process until the invariant finally 
reaches the desired property. In other words, it gradually checks whether there exist states 
in the system falsifying the property, and if so the algorithm returns false, otherwise 
returns true after exploring the whole system. Thus IC3 has the advantage that it checks also 
unsatisfiable cases.  (Note that both IMC and IC3 were proposed for checking safety properties 
first, and afterwards adapted to check the liveness ones by using the techniques shown in \cite{BAS02} and \cite{BSHZ11}.)

Compared to IMC and IC3, ours approach uses explicitly formula expanding instead of over-approximate 
inductive invariants. So the states information are stored during our checking process, easily making
it complete. In our approach, SAT invoking occurs when we extract the obligation formula (or its variant) 
from each expanded state. Hence our approach is hybrid, combining the explicit and symbolic ways together. 
Moreover, our approach applies for the whole LTL class, while IMC and IC3 are directly applicable
only for safety checking, as mentioned above. In this paper we compare our method with all LTL 
satisfiable-checking algorithms appearing in \cite{SD11} and find it is quite efficient. We leave to
future work to set up a comprehensive comparison with the full SAT-based model checking approaches 
such as IMC and IC3. And the comparison on SAT solver invoking times among different solvers will 
also be involved, which is well concerned in SAT community.

\section{Conclusion}\label{sec:con}

In this paper we proposed a fully-SAT-based LTL-satisfiability-checking approach. Our experiments show the new method significantly improves
the performance of Polsat, a portfolio-based LTL satisfiability solver. Thus, we believe that SAT-based LTL satisfiability checking has a promising future.

\section{Acknowledgement}
The authors are thankful for valuable comments and tool guidance from Rajeev Gor\'e, Ullrich Hustadt and Victor Schuppan. 

\bibliographystyle{plain}
\bibliography{ok,cav}

\end{document}